\documentclass[letterpaper]{IEEEtran}

\IEEEoverridecommandlockouts
\usepackage{cite}
\usepackage{amsmath,amssymb,amsfonts}
\usepackage{amsthm}

\newtheorem{lemma}{Lemma}
\newtheorem{theorem}{Theorem}

\theoremstyle{remark}
\newtheorem*{remark*}{Remark} % 星号=无编号

\usepackage{algorithmic}
\usepackage{graphicx}
\usepackage{textcomp}
\usepackage{xcolor}
\newcommand{\rev}[1]{\textcolor{black}{#1}}
\usepackage[caption=false,font=footnotesize]{subfig}

\def\BibTeX{{\rm B\kern-.05em{\sc i\kern-.025em b}\kern-.08em  T\kern-.1667em\lower.7ex\hbox{E}\kern-.125emX}}

\newcommand{\PVtot}{PV^{\text{tot}}_j}  
\begin{document}

\title{\LARGE{ISI Modeling and BER Performance for Rotating Light-Trail Image Sensor Communication}}

%\author{Author 1, Author 2, Author 3, and~Author 4}
\vspace{-4mm}
\author{Shin Asaoka,
        Shan Lu,~\IEEEmembership{Member,~IEEE,}
        Zhengqiang Tang,
        and~Takaya Yamazato,~\IEEEmembership{Senior Member,~IEEE}
        % <-this % stops a space
 \thanks{Received 23 April 2025; accepted 4 May 2025. The associate editor is Dr. Sotiris Tegos. This work was supported by JSPS KAKENHI Grant Number 25K00371 and Telecommunications Advancement Foundation, Japan. (Corresponding author: Shan Lu.)}
\thanks{S. Asaoka, S. Lu and T. Yamazato are with Department of Information and Communication Engineering, Graduate School of Engineering,
 Nagoya University, Nagoya, 464-8601, Japan. (Email: shan.lu.jp@ieee.org.)}
\thanks{Zhengqiang Tang are with Dept of Elec. and Elec. Eng., Shizuoka University, Shizuoka Japan.}
}

% The paper headers
\markboth{IEEE communication letters,~Vol.~, No.~, ~2026}%
{Shell \MakeLowercase{\textit{et al.}}: Bare Demo of IEEEtran.cls for IEEE Communications Society Journals}

\maketitle

\begin{abstract}
Image sensor communication (ISC) employing a propeller-LED transmitter encodes data along rotating light trails. We present an analytical framework that (i) constructs a single-LED, single-blink light trail model that maps optical power to pixel values, and (ii) integrates a probabilistic noise model to derive a closed-form bit-error rate (BER) using the $Q$-function. Trimodal pixel-value histograms motivate an adjacent-only inter-symbol interference (ISI) model in which the decision at segment $j$ depends on adjacent segments. Applying a hardest-pair midpoint threshold yields per-segment BER and a general BER after marginalization. We further provide practical sufficiency conditions under which adjacent-only ISI is adequate, and validate its tightness against Monte Carlo simulations and experimental results. Using the analytical BER, we select the control angle that maximizes throughput while satisfying a target BER reliability constraint.
\end{abstract}

\vspace{-4mm}
\section{Introduction}
Image sensor communication (ISC) is a visible light communication (VLC) technique that uses an image sensor (camera) as the receiver~\cite{komine,wook}. In ISC, the camera records the on–off states of light-emitting diodes (LEDs) as two-dimensional (2-D) image frames~\cite{kamakura}. 
Owing to its spatial parallelism and selectivity, which enable interference rejection via simple image processing~\cite{nagura,arai}, together with low cost and spatial privacy, ISC is attractive for short-range IoT links and screen-to-camera applications. Unmanned aerial vehicles (UAVs, drones) are increasingly deployed for disaster response and military missions, where radio-frequency links may be congested, regulated, or vulnerable to interception; as a result, VLC on drones has attracted growing interest~\cite{motwani,Takano}.

However, with static LEDs and frame-synchronous blinking, each on–off cycle is captured at most once per frame. Consequently, the per-LED symbol rate is fundamentally limited by the camera frame rate, and the per-frame throughput is further constrained by the number of LEDs within the field of view.

Leveraging a quadrotor’s spinning propellers, a propeller-LED transmitter (P-Tx) mounts LEDs on a blade and modulates based on the rotation angle to increase the rate of ISC~\cite{tang2022}. During exposure, the rotating emission integrates into an arcuate “light trail” that acts as a virtual extended source, enabling angle-mapped modulation without additional optics or tight synchronization and thereby increasing throughput. 
By partitioning the $2\pi$ angular range into $J$ sectors and mapping symbols to these angular segments, the P-Tx embeds $J$ parallel pieces of information within a single image frame, achieving high-capacity transmission with only a small number of LEDs~\cite{arai2}.
% \rev{The key advantage of this architecture is that it breaks the traditional frame-rate bottleneck, mapping temporal modulation into spatial angles, thus enabling high-density communication.}
\rev{The key advantage is breaking the frame-rate bottleneck, enabling high-density communication. Thus, we position this as a low-rate, high-robustness method for scenarios like disaster monitoring and RF-restricted environments.}
While P-Tx VLC systems have been validated experimentally, theoretical analysis is scarce. Key parameters likes rotation angle $\Delta\theta = 2\pi/J$ have been chosen empirically to balance throughput and reliability~\cite{tang2022}, leaving BER behavior and design suboptimal.

In this paper, we provide an end-to-end analysis and design framework for P-Tx-based ISC. The main contributions are:

(1) \textbf{Single-LED single-blink model:} We develop a single-LED single-blinking light trail model that maps instantaneous optical transmit power of a single LED to the pixel values through the LOS channel, defocus blur, and the camera. 

(2) \textbf{Adjacent-only ISI modeling and closed-form BER:} Guided by the observed multi-peak pixel-value histograms, we formulate an adjacent-only ISI model for light trail segments. With a \emph{hardest-pair midpoint} decision threshold, we derive a closed-form BER expression in terms of the 
$Q$-function and validate its tightness against Monte Carlo simulations and experiments.

(3) \textbf{Design of the control angle $\Delta\theta$:}
Using the analytic BER, we optimize $\Delta\theta$ under a target-BER (reliability) constraint and demonstrate throughput maximization consistent with optimized $\Delta\theta$.

\vspace{-3mm}
\section{System model: P-Tx-based ISC}\label{sec:systemmodel}
This section presents the P-Tx-based ISC system with light trails, which consists of a P-Tx, an optical channel, and a receiver (camera).
\vspace{-4mm}

\begin{figure}[htbp]
    \centering
    \subfloat[Static P-Tx.]{
        \includegraphics[width=0.15\textwidth]{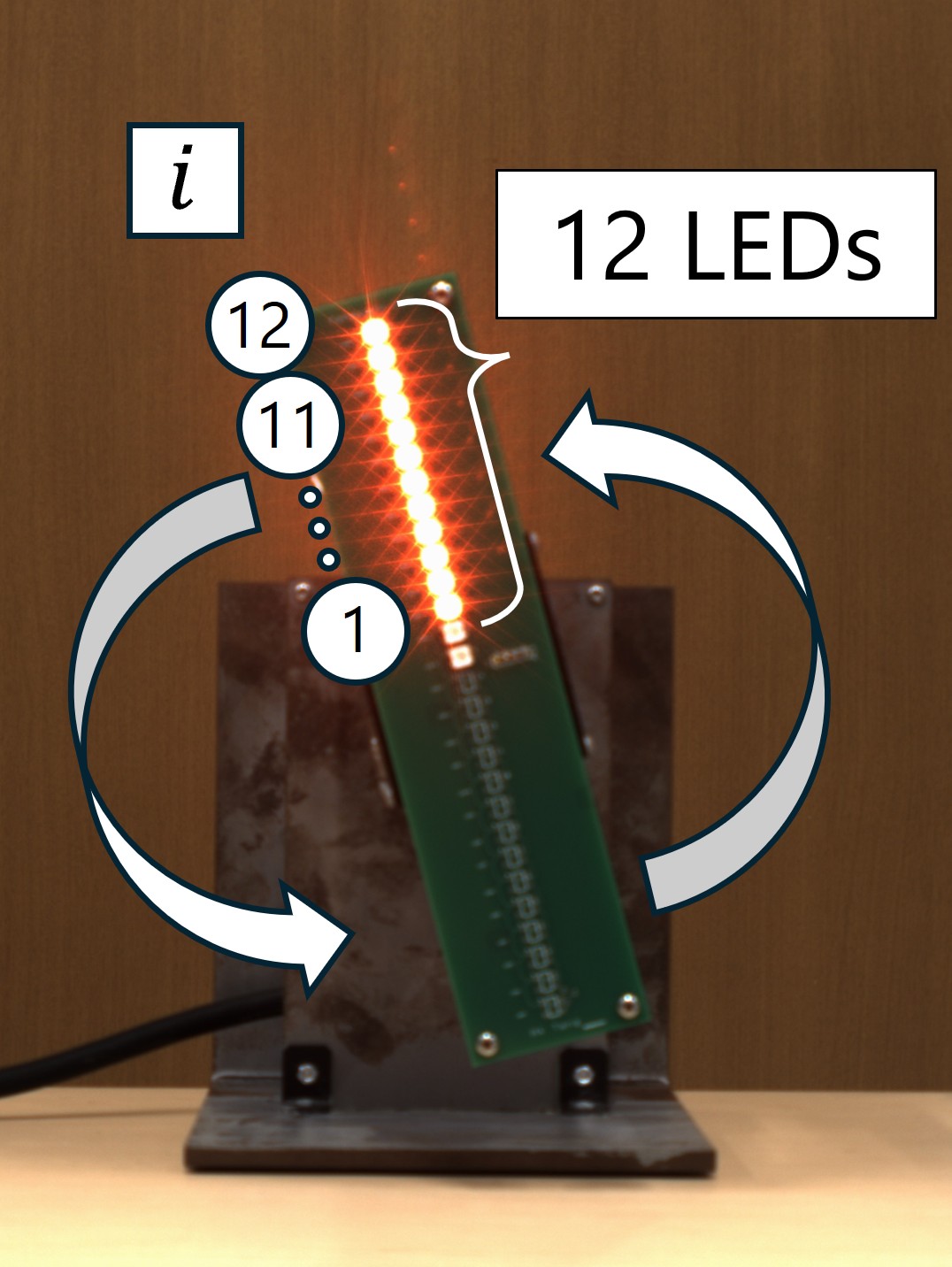}
    }
    \hfill
    \subfloat[Light Trails.]{
        \includegraphics[width=0.15\textwidth]{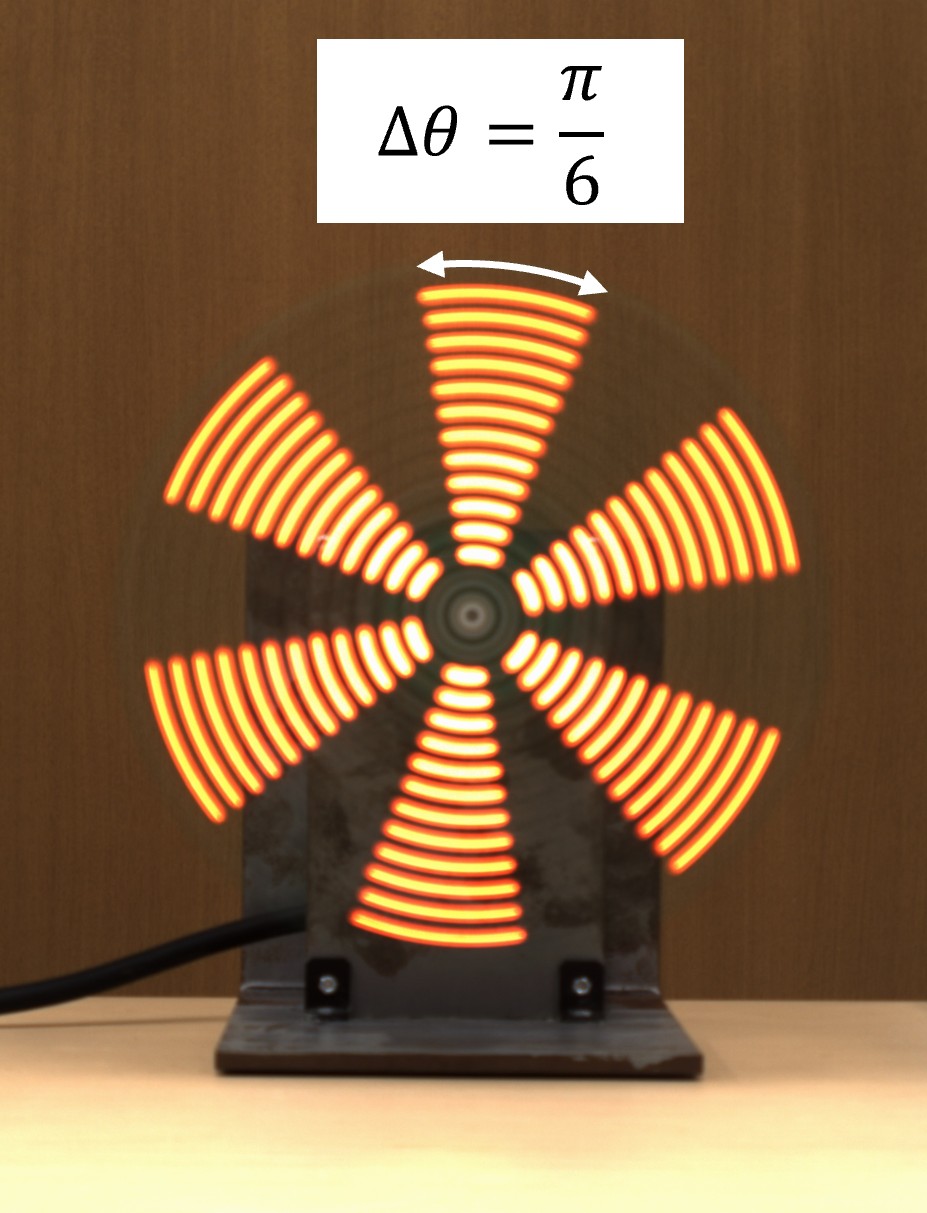}
    }
    \caption{Propeller LED Transmitter (P-Tx).}
    \label{fig:P-Tx}
    \vspace{-5mm}
\end{figure}

 \vspace{-2mm}
\subsection{P-Tx transmitter}
We use a P-Tx, shown in Fig.~\ref{fig:P-Tx}. Twelve LEDs are arranged in a linear array with a spacing of 7-mm.
For the $i$-th LED ($i=1,2,\dots,12$) at radius $r_{i}$ from the rotation axis, the blinking state is controlled at each rotation angle $\Delta\theta_i$ and the information is transmitted using On-Off Keying (OOK). Each full rotation has $J_i=2\pi/\Delta\theta_i$ bits. 
The data transmitted for the $i$-th LED at the angular position $j\in\{0, \ldots,J_{i}-1\}$ is $b^i_j$.

We assume a parallel line-of-sight configuration in which the optical axes of the P-Tx and the camera are aligned. When the P-Tx rotation period is synchronized with the camera exposure time, the sensor integrates a full rotation into a single frame, producing a complete light trail~\cite{tang}.

\vspace{-4mm}
\subsection{Demodulation of the receiver}\label{subsec:montecarlo}
For any LED $i$, the emitted light forms an approximately circular trace in the image plane. We partition this trace into
$J_i$ angular segments and, for each segment, extract a representative pixel value.

Let the image resolution be $X\times Y$ with pixel coordinates $(x,y)$. Denote the centroid of the $j$-th segment for the LED $i$ by $(x^i_j,y^i_j)$. 
We define the segment sample $PV^i_{x_j,y_j}$ of the $j$-th segment by the pixel intensity at $(x^i_j,y^i_j)$. 

The detected bit is obtained by threshold $PV_{\text{th}}^i$ as
\begin{align}
\hat{b}^i_{j} =
\begin{cases}
  1 & \text{if } PV^i_{x_j,y_j} > PV_{\text{th}}^i \\
  0 & \text{if } PV^i_{x_j,y_j} \le PV_{\text{th}}^i.
\end{cases}
\end{align}
\vspace{-5mm}
\section{Single-LED single-blink Light Trail Model}
Since the pixel value is the fundamental observable for ISC demodulation, this section presents a
\textbf{\textit{single-LED
single-blink light trail model}} 
that maps the instantaneous optical transmit power of one LED to pixel values through the LOS channel, defocus blur, and the camera, as illustrated in Fig.~\ref{fig:systemmodel}.
Since we consider a single LED with single-blink, we \textbf{omit the index $i$ and $j$} in this section.

\begin{figure}[t] 
     \centering
     \includegraphics[width= 1.0\linewidth]{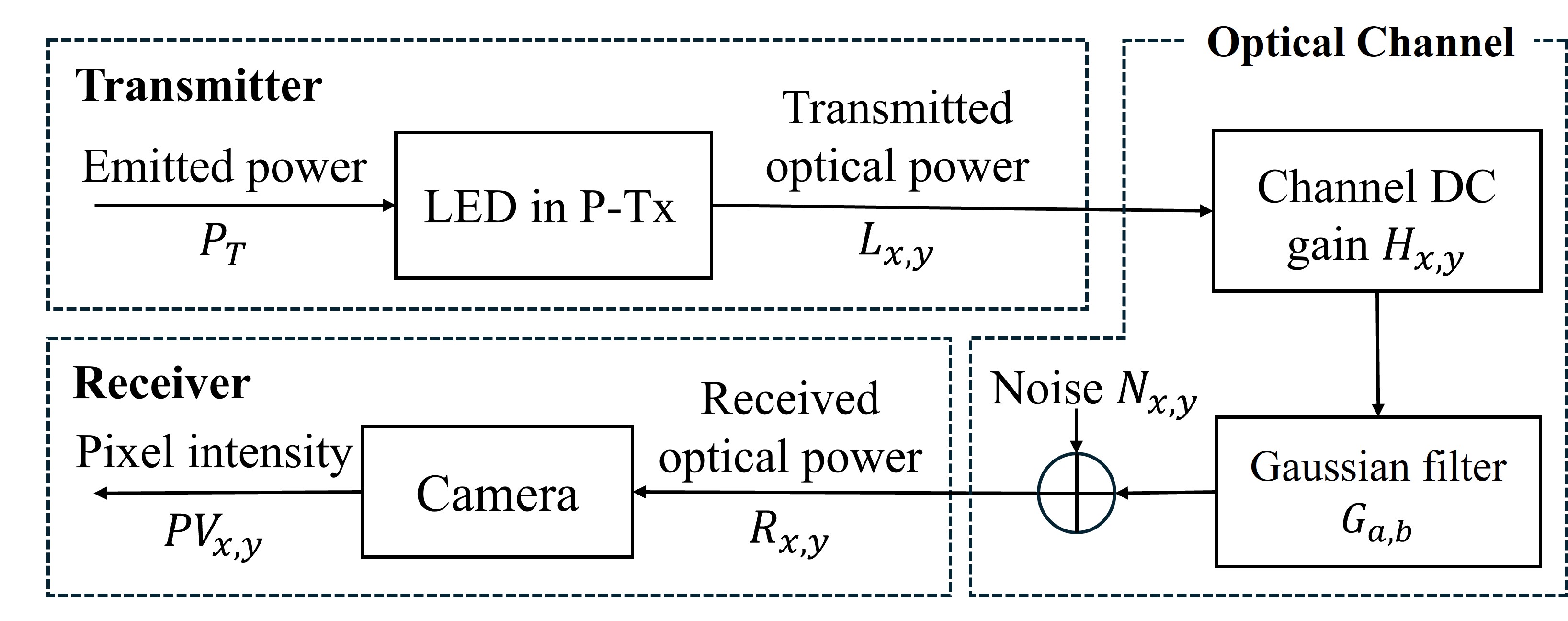}
     \caption{Light trail source–channel–camera pipeline used in the analytical model.}
     \label{fig:systemmodel}
     \vspace{-5mm}
 \end{figure}

\subsubsection{Transmitter}
Let $P_{\text{tot}}$
be the total \textit{emitted power} under constant illumination (i.e., when all $J$ segments transmit “1”). 
If $\tilde J$ out of $J$ segments are active, the emitted (transmitter-side) frame-level power is $P_T \;=\; \tilde J \cdot P_{\text{tot}}/J.$

Let $\mathbf{P}\in\mathbb{R}^{X\times Y}$
denote the (unnormalized) radiometric power distribution of the light trail on the sensor,
with entries $P_{x,y}$ obtained from the luminous-energy accumulation $Q^{\text{blink}}_{x,y}$ along the blinking trail~\cite{asaoka}:
\begin{equation}
P_{x,y} \;=\; \frac{Q^{\text{blink}}_{x,y}}{K\,V(\lambda)}.
\label{eq:Pxy}
\end{equation}
where $K$ is the maximum luminous efficacy coefficient and $V(\lambda)$ is the relative luminous efficiency.

Let $L_{x,y}$ denote the per-pixel incident optical energy at pixel $(x,y)$ during the exposure.
We allocate $P_T$ across pixels proportionally to $\mathbf{P}$ with 
\begin{equation}
L_{x,y} \;=\; P_T\,\frac{P_{x,y}}{\sum_{u,v} P_{u,v}}
\label{eq:Lxy}
\vspace{-2mm}
\end{equation}
which guarantees $\sum_{x,y} L_{x,y}=P_T$.

\subsubsection{Optical channel and defocus blur}
We adopt a distance-dependent LOS channel~\cite{liu}. The DC gain at pixel $(x,y)$ is
\begin{align}
H_{x,y} \;=\; \frac{A_p}{D_{x,y}^{\gamma}}\,R(\psi_{x,y}),
\label{eq:Hxylos}
\vspace{-3mm}
\end{align}
where $D_{x,y}$ is the source–pixel distance, $\gamma$ the path-loss exponent, and $R(\psi)$ is the Lambertian radiant pattern with order $m'$ as
$R(\psi)=\frac{m'+1}{2\pi}\cos^{m'}\!\psi.$

The effective collection area is
\begin{equation}
A_p \;=\;
\begin{cases}
A\,T_s\,g\,\cos\epsilon_{x,y}, & 0\le \epsilon_{x,y}\le \epsilon_\ell,\\
0, & \epsilon_{x,y}>\epsilon_\ell,
\end{cases}
\label{eq:Ap}
\end{equation}
with entrance pupil area $A$, filter transmittance $T_s$, lens gain $g$, angle of incidence $\epsilon_{x,y}$, and lens FoV $\epsilon_\ell$.

Defocus acts as a spatial low-pass filter that grows with distance~\cite{arai3}. 
We model it by a normalized discrete Gaussian kernel $\{G_{a,b}\}$ of size $q\times q$ and variance $\sigma_g^2$:
\begin{equation}
G_{a,b}\;=\;\frac{\exp\!\big(-(a^2+b^2)/(2\sigma_g^2)\big)}{\sum\limits_{u=-(q-1)/2}^{(q-1)/2}\sum\limits_{v=-(q-1)/2}^{(q-1)/2}\exp\!\big(-(u^2+v^2)/(2\sigma_g^2)\big)}.
\label{eq:GaussKernel}
\end{equation}
The received optical power after blur is the 2-D convolution of  per-pixel incident optical energy $L_{x,y}$ with $G$ scaled by $H$:
\begin{equation}
R_{x,y} \;=\; H_{x,y}\sum_{a,b} G_{a,b}\,L_{x-a,y-b}.
\label{eq:Rxy}
\end{equation}
\vspace{-3mm}
\subsubsection{Camera receiver}\label{subsec:receiver}
The photon count at pixel $(x,y)$ is % modeled as
\begin{equation}
I_{x,y} \;=\; \frac{R_{x,y}+N_{x,y}}{Q_p}, 
%\qquad \mu_{x,y}\triangleq \frac{R_{x,y}}{Q_p},
\label{eq:photon}
\vspace{-2mm}
\end{equation}
where $R_{x,y}$ is the expected radiant energy accumulated during exposure and $Q_p$ is the photon energy.
Here $N_{x,y}$ denotes the energy-domain fluctuation due to photon shot and readout noise.
Under the normal approximation to the Poisson term, $N_{x,y}$ is modeled as
$N_{x,y}\;\sim\;\mathcal N\!\big(0,\sigma^2_n\big)$.
% so that in the photon-count domain we have
% \begin{equation}
% I_{x,y}\;\approx\;\mu_{x,y}+n_{x,y}, 
% \qquad n_{x,y}\sim\mathcal N\!\big(0,\;\mu_{x,y}+\sigma_{r}^{2}\big),
% \quad \sigma_{r}^{2}\triangleq \sigma_{E,\mathrm{read}}^{2}/Q_p^{2}.
% \end{equation}

% The photon count at the pixel $(x,y)$ is modeled as~\cite{liu}
% \begin{equation}
% I_{x,y} \;=\; \frac{R_{x,y}+N_{x,y}}{Q_p}, \qquad
% N_{x,y}\sim\mathcal N\!\big(0,\sigma_n^2\big),
% \label{eq:photon}
% \end{equation}
% where $Q_p$ is the photon energy and $\sigma_n^2$ is the variance of the noise from the photon shot.
The camera pipeline (photoelectric conversion, amplification, A/D, gamma) is summarized by $\Phi(\cdot)$, and the pixel value at $(x,y)$ is
\vspace{-1mm}
\begin{align}
\label{eq:Phi}
PV_{x,y} &= \Phi\!\big(I_{x,y}\big) \\
&= t\!\left(\frac{A_{a}A_{c}\!\left(V_{a_{\text{ref}}}-A_{f}\!\left(V_{\text{ref}}-Q_{e}A_{n}I_{x,y}\right)\right)}{\mathrm{raw}_{\max}}\right)^{\!1/g_{a}}, \nonumber
\end{align}
where $t(\cdot)$ denotes clipping to $[0,1]$ and the parameters are listed in Table~\ref{tab:parameters}.%

\vspace{-2mm}

\section{ISI modeling and analytic BER evaluation}
\subsection{Motivation for ISI Model}\label{subsec:histogram_analysis}
Having introduced the single-blink light trail model, we now consider a single LED with multiple blinks. Fig.~\ref{fig:simulation_hist} shows the pixel intensity histogram for $\Delta\theta=\pi/9$ and a 52 m link.
The light trail is produced by a single LED that blinks on/off randomly using OOK (5,000 random bits).
For both symbol classes (0 and 1), the histogram is trimodal, with three peaks aligned with the adjacent-segment states (00), (01/10), and (11).
% ここで、ヒストグラムにおける各分布の広がりは、式8で示した雑音や式7におけるボケの影響、および連続的な光跡を離散的な2次元のピクセル分布にマッピングする際に生じる各光跡セグメントの幾何学的な位置関係誤差に伴う画素値の変動によって生じる。
\rev{The spread of each distribution in the histogram is caused by the blurring effect in Eq.~\eqref{eq:GaussKernel}, the noise shown in Eq.~\eqref{eq:photon}, and pixel value fluctuations associated with geometric positional errors of each light trail segment that arise when mapping continuous light trails onto a discrete 2D pixel distribution.}
This empirical pattern is consistent with inter-symbol interference (ISI) from the two adjacent light trail segments and motivates the ISI model developed below.
\begin{figure}[t] 
\centerline{\includegraphics[width= 0.9\linewidth]{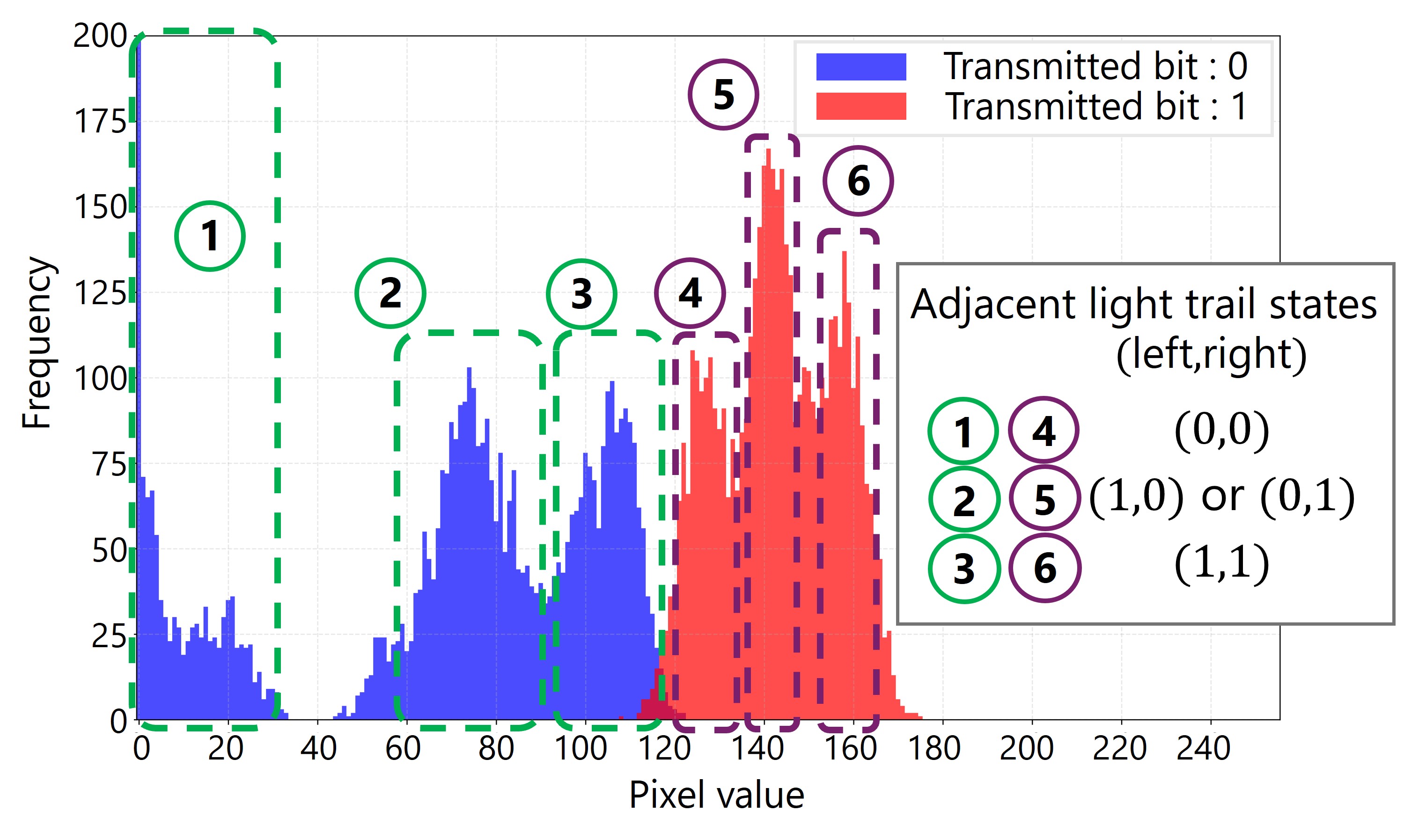}}
     \caption{Histogram of received pixel values in the simulation.}
     \label{fig:simulation_hist}
     \vspace{-5mm}
 \end{figure} 
\vspace{-5mm}

\subsection{Theoretical BER Performance Analysis}\label{subsec:theoretical_analysis}

\subsubsection{BER analysis for interference model}
We consider an interference model in which the received power at $(x_j,y_j)$ is modeled as a superposition of the $j$-th segment and its first neighbors, while the tail energy beyond the first neighbors is negligible:
\vspace{-2mm}
\begin{align}
\rev{R_{x_j,y_j}}
=  b_{j-1}R_{x_j,y_j}^{j-1}
+ b_{j}R_{x_j,y_j}^{j}
+ b_{j+1}R_{x_j,y_j}^{j+1}.
\label{eq:R_total_k1}
\end{align}
where $b_{j}\in \{0,1\}$ is the on/off state of the $j$-th segment and the component powers $R_{x_j,y_j}^{j\pm1}$ are obtained from the physical model in Eq.~\eqref{eq:Rxy} by activating only segment
$j\pm1$ at $(x_j,y_j)$. 
% \rev{Note that, since Eq.~\eqref{eq:Rxy} is derived via 2D convolution over the curved trajectory, this formulation inherently captures the impact of the isotropic 2D Gaussian blur and the geometric asymmetry relative to the pixel grid.}
\rev{Hence, Eq.~(10) already reflects the effect of the isotropic Gaussian blur and the position-dependent geometric variation induced by pixel-grid alignment.}

% \begin{figure}[t] 
%      \centering
%      \includegraphics[width= 0.9\linewidth]{images/R_j.jpg}
%      \caption{Blinking pattern for calculating $R_{x_j,y_j}^j$. }
%      \label{fig:R_j}
%      % \vspace{-6mm}
%  \end{figure}

Following Eq.~(\ref{eq:photon}) in Sec.~\ref{subsec:receiver}, the camera maps the noisy received power to a pixel value via the response function $\Phi(\cdot)$.
Since $\Phi(\cdot)$ is nonlinear, additive power-domain noise becomes signal dependent in the pixel domain.
For small noise (i.e., around a local operating point), we linearize $\Phi(\cdot)$ at $\mu_j \triangleq R^{\text{ISI}}_{x_j,y_j}/Q_p$ using a first-order Taylor expansion as:
\vspace{-3mm}
\begin{align}
PV_{j}^{\text{re}}
&= \Phi\left((\rev{R_{x_j,y_j}} + N_{x_j,y_j})/Q_p\right) \nonumber \\
%&  \approx \Phi\left(\rev{R_{x_j,y_j}}/Q_p\right) + N^{\prime}_{x_j,y_j}(b_{j-1},b_j,b_{j+1}) \\
&\approx \Phi\!\left(\frac{\rev{R_{x_j,y_j}}}{Q_p}\right) +\Phi'\!\left(\frac{\rev{R_{x_j,y_j}}}{Q_p}\right)\frac{N_{x_j,y_j}}{Q_p}. \nonumber
%&\triangleq \PVtot(b_{j-1},b_j,b_{j+1}) + N^{\prime}_{x_j,y_j}(b_{j-1},b_j,b_{j+1}).  \nonumber
\label{eq:PV_recv}    
\end{align}

When $(b_{j-1},b_j,b_{j+1})$ is fixed, $\Phi (\rev{R_{x_j,y_j}})$ is also fixed. 
Therefore, we define averaged value of $PV_{j}^{\text{re}}$ and the equivalent pixel-domain noise as
\vspace{-2mm}
$$
\PVtot(b_{j-1},b_j,b_{j+1}) \triangleq \Phi (\rev{R_{x_j,y_j}}) 
$$
\vspace{-3mm}
$$
N'_{x_j,y_j}(b_{j-1},b_j,b_{j+1})
\;\triangleq\;
\Phi'\!\left(\frac{\rev{R_{x_j,y_j}}}{Q_p}\right)\frac{N_{x_j,y_j}}{Q_p}.
$$

Since $N_{x_j,y_j}\sim\mathcal N(0,\sigma_n^2)$, the term above is zero-mean Gaussian with variance
\vspace{-3mm}
\begin{align}
\sigma_n'^2(b_{j-1},b_j,b_{j+1})
&=\left(\frac{\Phi'\!\left(\rev{R_{x_j,y_j}}/Q_p\right)}{Q_p}\right)^{\!2}\sigma_n^2 .
%\label{eq:sigma_prime}
\end{align}
Hence, the received pixel value is Gaussian around the ideal (noise-free) response:
\begin{align}
&PV_{j}^{\text{re}}(b_{j-1},b_j,b_{j+1}) \nonumber \\
&\sim
\mathcal N\!\Big(\PVtot(b_{j-1},b_j,b_{j+1}),\; \sigma_n'^2(b_{j-1},b_j,b_{j+1})\Big).
\label{eq:PV_gauss}
\end{align}

% In our experiments, the operating point is kept in the mid-range, so the dependence of $\sigma_n'^2(b_{j-1},b_j,b_{j+1})$ on the bit triplet is weak. Unless otherwise stated, we therefore adopt a \textbf{\textit{locally homoscedastic}} approximation and use a single \textbf{$\sigma_n'^2$} per working condition, which we estimate empirically from repeated frames.

\rev{In general, the equivalent pixel-domain noise variance depends on the neighboring-bit pattern through the nonlinear camera response. In our experiments, the received pixel values lie in a mid-range region where the variation of $\Phi'(\mu_j)$ across dominant triplets is small. We therefore use a \textit{locally homoscedastic approximation} with one effective variance $\sigma_n'^2$ per working condition for analytical tractability. For stronger nonlinearity or larger signal excursions, a triplet-dependent variance model would be more appropriate.}

\vspace{-2mm}
\begin{lemma}[Monotonicity $\Rightarrow$ worst-case neighbors]
Let the detector use a fixed threshold $PV_{\text{th}}$ and 
Gaussianized noise with standard deviation $\sigma_n'$ independent of the bits, and nonnegative adjacent contributions $R^{j-1}_{x_j,y_j},R^{j+1}_{x_j,y_j}\ge 0$. 
Define 
$$\mu_i(b_{j-1},b_{j+1}) \triangleq  \PVtot(b_{j-1},i,b_{j+1}), ~~ i=0,1.$$
Then, conditional error probabilities are monotone in each neighbor bit:
$$
P(\hat b_j=e \mid b_{j-1}, b_j=i, b_{j+1})
 = Q\!\Big(\tfrac{PV_{\text{th}}-\mu_i(b_{j-1},b_{j+1}) }{\sigma_n'}\Big) . $$
where $Q(\cdot)$ is strictly decreasing. Hence, the worst-case neighbor patterns are
$(b_{j-1},b_{j+1})=(1,1)$ for $b_j=0$ and $(0,0)$ for $b_j=1$.
\end{lemma}

\begin{proof}
We show monotonicity in $b_{j-1}$; the argument for $b_{j+1}$ is identical.

\textbf{Case $b_j=0$.}
Fix $b_{j+1}$. When $b_{j-1}$ flips from $0$ to $1$, the conditional mean
$\mu_0$ \textbf{increases }by $R_{x_j,y_j}^{j-1}\ge 0$ while $\sigma_n'^2$ and $PV_{\text{th}}$ stay fixed.
Because $Q(x)$ is strictly decreasing in $x$,
\[
Q\!\Big(\tfrac{PV_{\text{th}}-\mu_0(1,b_{j+1})}{\sigma_n'}\Big)
\;>\; 
Q\!\Big(\tfrac{PV_{\text{th}}-\mu_0(0,b_{j+1})}{\sigma_n'}\Big).
\]
Therefore $P(\hat b_j=1\mid b_{j-1},b_j=0,b_{j+1})$ \emph{\textbf{increases}} with $b_{j-1}$.
By symmetry in $b_{j+1}$, the error is maximized at $(b_{j-1},b_{j+1})=(1,1)$.

\textbf{Case $b_j=1$.}
Fix $b_{j+1}$. When $b_{j-1}$ flips from $0$ to $1$, the conditional mean
$\mu_1$ \textbf{increases }by $R_{x_j,y_j}^{j-1}\ge 0$.
Using the same monotonicity of $Q(\cdot)$,
\[
Q\!\Big(\tfrac{\mu_1(1,b_{j+1})-PV_{\text{th}}}{\sigma_n'}\Big)
\;<\; 
Q\!\Big(\tfrac{\mu_1(0,b_{j+1})-PV_{\text{th}}}{\sigma_n'}\Big).
\]
Hence $P(\hat b_j=0\mid b_{j-1},b_j=1,b_{j+1})$ \emph{\textbf{decreases}} with $b_{j-1}$.
By symmetry, the error is maximized at $(b_{j-1},b_{j+1})=(0,0)$.

Combining both cases yields the stated worst-case neighbor patterns.
\end{proof}

\vspace{-3mm}
\begin{remark*}[\textbf{The midpoint threshold}]
The ``worst-neighbor'' ordering and the resulting bounds require a fixed threshold, a location-family statistic $X=\mu(b_{j-1},b{j}, b_{j+1})+n$ with
nonnegative neighbor shifts (flipping $0\!\to\!1$ does not decrease the mean), and noise with a
monotonically decreasing tail.
With approximately equal variances for $b_j\in\{0,1\}$ and equal priors, the Bayes threshold for the hardest
pair $(1,0,1)$ vs.\ $(0,1,0)$ is the midpoint 
\begin{equation}
PV_{\text{th}}
=\tfrac{1}{2}\,\PVtot(1,0,1) +\tfrac{1}{2}\,\PVtot(0,1,0).
\label{eq:PV_threshold}
\end{equation}
Otherwise, use the LLR-root threshold.
\rev{We adopt a single-variance approximation, leading naturally to the midpoint threshold as a practical decision boundary. This simplification provides analytical tractability. A detailed performance-gap quantification is left for future work.}
\end{remark*}
% \begin{remark*}[\textbf{General rule under unequal variances}]
% Consider a location-family statistic $X=\mu(b_{j-1},b_{j+1})+n$ with nonnegative neighbor shifts
% (flipping $0\!\to\!1$ does not decrease the mean) and noise with a monotonically decreasing tail
% (log-concave MLR family, e.g., Gaussian/Laplace). Under these mild conditions, the Bayes decision
% for the hardest pair $(1,0,1)$ vs.\ $(0,1,0)$ uses the \emph{LLR-root threshold} $PV_{\text{th}}$,
% defined as the solution to $\Lambda(\eta)=0$ with $\Lambda(x)\!\triangleq\!\log f_1(x)-\log f_0(x)$.
% For Gaussian $n$ with unequal variances $\sigma_0^2\neq\sigma_1^2$, $\eta^\star$ solves
% \[
% \frac{(x-\mu_1)^2}{2\sigma_1^2}-\frac{(x-\mu_0)^2}{2\sigma_0^2}+\log\!\frac{\sigma_1}{\sigma_0}=0.
% \]
% \end{remark*}

\begin{theorem}[Analytic per-segment BER under adjacent-only ISI with worst-case midpoint threshold]
\label{thm:adjacent-ber}
Consider the per-segment decision at $(x_j,y_j)$ under the adjacent-only ISI model with nonnegative taps and
Gaussianized noise $n\sim\mathcal N(0,\sigma_n'^2)$ independent of the bits. Using the fixed threshold
$PV_{\text{th}}$ in \eqref{eq:PV_threshold}, 
the conditional decision statistic for a given triplet
$(b_{j-1},b_j=i,b_{j+1})$ is Gaussian with mean
$\PVtot(b_{j-1},i,b_{j+1})$
and variance $\sigma_n'^2$. The conditional error probability is
\begin{align}
P(\hat{b}_j \not= i|b_{j-1}, i,b_{j+1}) = 
Q\!\left(
\frac{\,|PV_{\text{th}}-\PVtot(b_{j-1},i,b_{j+1})|\,}{\sigma_n'}
\right).
\label{eq:cond_err}
\end{align}
\end{theorem}
\vspace{-3mm}
In our system, a light trail forms a continuous arc in physical space but is sampled on the discrete pixel grid of the image sensor. This pixel-grid sampling induces a geometric asymmetry that depends on the relative alignment between the trail center and the pixel lattice. As a result, the shape and area of the pixel region covered by a given segment can vary slightly with position, introducing a mild position-dependent bias in the pixel value at each segment centroid. To mitigate this discretization bias, we evaluate the overall BER as the average of the per-segment BERs across all $J$ segments.

% Marginalizing the conditional error in \eqref{eq:cond_err} over the priors of the current bit and its two neighbors across $J$ segments yields the following general BER expression:

Marginalizing the conditional error in \eqref{eq:cond_err} over the priors of the current bit and its two neighbors, \rev{and averaging over all segments $j$ to capture residual discretization differences,} yields the following BER expression:
\begin{align}
\label{eq:ber_general}
\mathrm{BER} &=  
\frac{1}{J}\sum_{j=0}^{J-1}\;
\sum_{i\in\{0,1\}}\pi_i\;
\sum_{(b_{j-1},b_{j+1})\in\{0,1\}^2}\!\!  \\
& \pi_{\text{nb}}(b_{j-1},b_{j+1})\;Q\!\left(
\frac{\big|\,PV_{\text{th}}-\PVtot(b_{j-1},i,b_{j+1})\,\big|}
{\sigma_n'}
\right),   \nonumber 
\end{align}
where $\pi_i = P(b_j = i)$ and $\pi_{\text{nb}}(b_{j-1}, b_{j+1}) = P(b_{j-1}, b_{j+1})$.

\begin{remark*}
\rev{\textbf{(When is adjacent-only ISI enough?)}}\\
\rev{Adjacent-only ISI is adequate when the \emph{leakage ratio}
$$\Lambda_2 \triangleq
\sum_{|m|\ge 2} R^{j+m}_{x_j,y_j}/\sum_{|m|=1} R^{j+m}_{x_j,y_j}$$
satisfies \(\Lambda_2 \le \varepsilon\) for a small tolerance (e.g., \(\varepsilon=0.1\)).
This means that the energy leaked from non-adjacent trails is small compared with that from the immediate neighbors.}

\rev{
In practice, \(\Lambda_2\) stays small whenever the \emph{tangential spacing}
\(S=r_i\Delta\theta\) is sufficiently larger than the effective blur scale \(\sigma_{\mathrm{eff}}\).
A convenient rule of thumb is
\(S/\sigma_{\mathrm{eff}} \gtrsim 1.5\),
under which the energy contribution from non-adjacent segments becomes negligible compared with that from adjacent ones. This condition typically corresponds to the regime where (i) the per-bit angular spacing is not too small (e.g., \(r_i\Delta\theta \gtrsim \tfrac{1}{2}r_{\text{LED}}\)), and (ii) the defocus blur does not spread appreciably beyond the nearest trails. If these conditions are violated (small \(\Delta\theta\), large distance, or heavy blur), a generalized \(K\)-neighbor model should be used.
}
\end{remark*}

% Two geometry/optics-based conditions help keep $\Lambda_2$ small:
% (i) the per-bit angular spacing is not too small (e.g., $r_i\Delta\theta \gtrsim \tfrac{1}{2}r_{\text{LED}}$), and
% (ii) defocus is moderate relative to trail spacing (for the Gaussian kernel in \eqref{eq:GaussKernel}, $3\sigma_g \lesssim r_i\Delta\theta$, so the effective support does not reach next-nearest trails).
% If either condition is violated, adopt a generalized $K$-neighbor model.

% \begin{remark}[$K$-neighbor interference model]
% In the following conditions, interference from non-adjacent light trails becomes non-negligible and the assumption that interference from adjacent light trails is the dominant factor may no longer be valid. When these conditions occur individually or in combination, a generalized $K$-neighbor interference model is necessary.
% \begin{enumerate}
% \item Extremely small control angle per bit $\Delta\theta$. If the center-to-center arc length between adjacent trails, $r_{i}\Delta\theta$, is less than half the LED-chip radius $r_{\text{LED}}$ (i.e., $r_{i}\Delta\theta \leq \frac{1}{2}r_{\text{LED}}$), contributions from the next-nearest segments must be considered.
% \item Significant defocus blur. A larger spatial low-pass filter
% the size $q$ of the spatial low-pass filter in Eq.~\eqref{eq:G2} spreads energy over more pixels, allowing non-adjacent trails to influence the target pixel.
% \item Long communication distances. Multiple light trails may project onto a single pixel, aggregating contributions beyond the immediate neighbors.
% \end{enumerate}
% \end{remark}

\begin{table}[t]
\centering 
\vspace{-1mm}
\caption{Simulation parameters.}
\label{tab:parameters}
\begin{tabular}{c|c}
\hline
\multicolumn{2}{c}{\textbf{Transmitter}} \\
% \hline
% LED specifications & ws2812b \\
\hline
Radius of LED chip ($r_{\text{LED}}$)  & $2.0\times 10^{-3}$ m \\
\hline
Rotation radius ($r_{i}$/m)  & $17.5\times 10^{-3}$ - $94.5\times 10^{-3}$\\
\hline
Control angle per bit ($\Delta\theta$) & $ \frac{\pi}{29}$ - $\frac{\pi}{4}$ rad \\
\hline
Total LED transmission power ($P_{\text{tot}}$) & 0.2 W \\
% \hline
% Rotation speed of P-Tx & $6\pi$ rad/s \\
%\hline
%Blinking pattern ($B$) & alternately \\
\hline
\multicolumn{2}{c}{\textbf{Channel}} \\
\hline
Transmission distance ($D$) & 46.0 - 62.0 m \\
\hline
Optical path loss coefficient ($\gamma$) & 2.0 \\
\hline
Transmittance of the optical filter ($T_s$) & 0.9 \\
\hline
FOV ($\epsilon_l$), Lens gain ($g$) & 15°, 1.0  \\
\hline
Order of the Lambertian emission ($m'$) & 1.0  \\
\hline
S.D. of Gaussian filter ($\sigma_g$) & 1.0-1.5\\
\hline
Filter Size ($q$) & 5 \\
\hline
S.D. of camera noise ($\sigma_n'$) & 4.065 \\
\hline
\multicolumn{2}{c}{\textbf{Receiver}} \\
% \hline
% Camera & U3-3890SE-C-HQ \\
% \hline
% Camera frame rate & 3 fps \\
% \hline
% Exposure time & 333,320 $\mu$s \\
% \hline
% Aperture & F8 \\
\hline
Image sensor resolution & $4000 \times 3000$ pixels \\
\hline
Image sensor pixel pitch ($\rho$) & $1.85 \times 10^{-6}$ m \\
\hline
Camera focal length ($f$) & $30\times 10^{-3}$ m \\
\hline
Quantum efficiency ($Q_e$) & 0.5 \\
\hline
Reference voltage ($V_{\text{ref}}$) & 3.1 V \\
\hline
Sense node gain ($A_n$) & $2.8 \times 10^{-4}$ V/e \\
\hline
Source follower gain ($A_f$) & 1.0 \\
\hline
Maximum ADC voltage ($V_{\text{a\_ref}}$) & 2.5 V \\
\hline
ADC gain ($A_a$), CDS gain ($A_c$) & 1.0, 1.0  \\
\hline
Maximum RAW value ($raw_{\text{max}}$) & 4095 \\
\hline
Gamma value ($g_a$) & 2.2 \\
\hline
\end{tabular}
\vspace{-4mm}
\end{table}

\vspace{-5mm}
\section{Simulation Results}
\subsection{Experimental Setup and Parameter Estimation}
Table~\ref{tab:parameters} lists the simulation/experimental settings.
Per distance, the pixel-domain noise level is estimated from real images via \emph{maximum-brightness extraction}~\cite{liu}: we record 557 frames under a fixed blink, take the framewise maxima, apply GAT+whitening, and use their variance as $\sigma_n'^2$.
Fig.~\ref{fig:simulation_hist} indicates near-homoscedasticity across $b_j\!\in\!\{0,1\}$, so we use a single $\sigma_n'^2$ per distance.
\rev{The operating points lie in the mid-range linear region of the camera response, ensuring the validity of the local linearization and near-homoscedastic approximation.} Blur parameters were fit by image–simulation matching.
\vspace{-4mm}

% Fig.~\ref{fig:simulation_hist} indicates near-homoscedasticity across $b_j\!\in\!\{0,1\}$, so we use a single $\sigma_n'^2$ per distance. Blur parameters were fit by image–simulation matching. \rev{We ensure that the system operates within the mid-range linear region of the camera response by adjusting exposure and selecting operating points accordingly.}
% For each distance we least-squares fit a Gaussian Point Spread Function (PSF) to the measured LED-trail radial profiles under the same blink/geometry, yielding $(\sigma_g, q)$.

% Table~\ref{tab:parameters} summarizes the simulation parameters. To estimate the camera noise level in the pixel domain, we applied the \emph{maximum-brightness extraction} procedure to experimentally captured images~\cite{liu}.
% For each link distance, we captured $557$ frames under a fixed blinking pattern. From each frame we extracted the per-frame maximum pixel value, and the resulting sequence of maxima was used to estimate the pixel-domain noise variance, yielding $\sigma_n'$ for that distance. 

% We use “Gaussianized” to indicate that images are variance-stabilized (GAT) and whitened before modeling. After GAT+whitening, the class-conditioned histograms in Fig.~4 show nearly identical standard deviations for $b_j\in\{0,1\}$, so we adopt a single shared pixel-domain noise level $\sigma_n'^2$ per distance.

\subsection{Simulation Results}
\subsubsection{Validation of the Analytic BER Model}
Fig.~\ref{fig:ber} compares the BER against Monte Carlo simulations and experimental measurements as a function of communication distance from $46~$m to $62~$m at $\Delta\theta=\pi/9$.

Our adjacent-only ISI model closely tracks the experimental BER and Monte Carlo BER over all tested distances. 
\rev{Incorporating the next-nearest neighbors (2-neighbor model) and the \textbf{all-segment} model, which considers the entire light trail, do not generally improve accuracy.}
In contrast, the “No ISI theoretical” curve assumes isolated segments with no inter-symbol interference, substantially underestimates the observed BER. 
Overall, the results confirm that adjacent-only ISI dominates the error behavior in this system and is sufficient for accurate BER prediction and link design. \rev{In tested conditions, the numerical results also indicate that the BER difference between the midpoint and the optimal threshold is negligible.}

\begin{figure}[t] 
\vspace{-2mm}
 \centering
 \includegraphics[width= 0.8\linewidth]{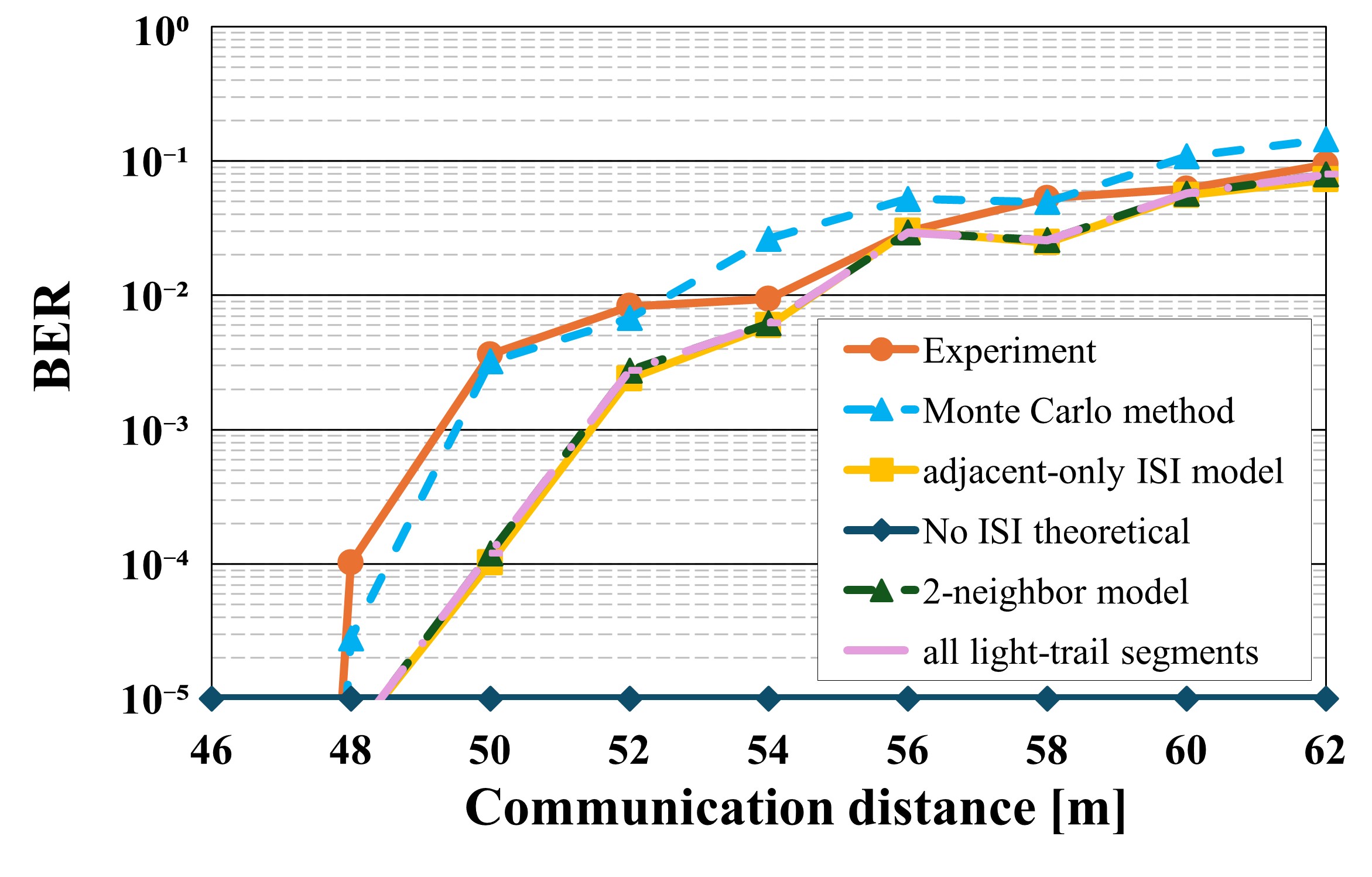}
  \vspace{-4mm}
 \caption{BER vs. communication distance ($\Delta\theta=\pi/9$).}
 \label{fig:ber}
\end{figure} 

\begin{figure}[t] 
 \vspace{-4mm}
 \centering
\includegraphics[width= 0.8\linewidth]{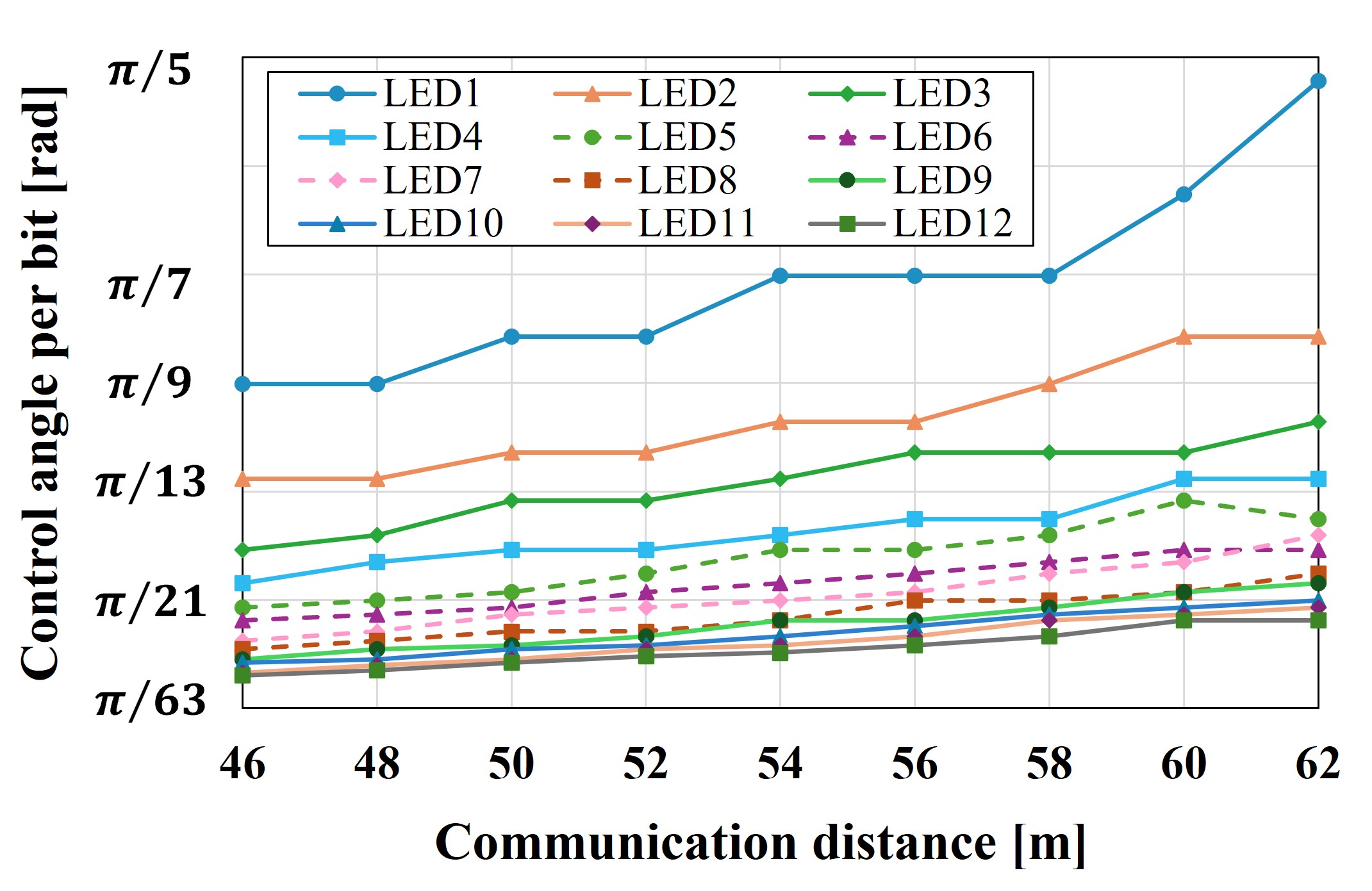}
\vspace{-3mm}
 \caption{Control angle ($\Delta\theta$) versus communication distance.}
 \label{fig:distancevsAngle}
 \vspace{-3mm}
\end{figure} 

\subsubsection{Optimization of control angle}
The control angle $\Delta\theta$ governs the trade-off between the data rate and communication reliability: decreasing $\Delta\theta$ improves the data rate but reduces the spatial resolution per bit, making the signal more susceptible to ISI.
Therefore, we determine the optimal control angle $\Delta\theta$ under an error-free criterion, defined as $\text{BER}\leq 10^{-4}$.
In the optimization, we restrict $J$ to positive even integers ($J=2a$ with $a\in\mathbb{N}$), thus the vertical axis is discretized at $\pi/a$.

% Fig.~\ref{fig:distancevsAngle} illustrates the optimal control angle $\Delta\theta$ across various communication distances for all twelve LEDs in the P-Tx.
% First, for any given LED, the optimal $\Delta\theta$ increases as the communication distance grows.
% This is because at longer distances, the light trail segment corresponding to a single bit covers fewer sensor pixels, which makes the signal more susceptible to ISI.
% Second, at any fixed distance, outer LEDs (with a larger rotation radius) require a smaller $\Delta\theta$.
% This occurs because the physical light trail length per unit rotation angle is longer for outer LEDs, allowing each bit to span more pixels on the sensor and thereby mitigating the impact of ISI.

\rev{Fig.~\ref{fig:distancevsAngle} shows the optimal control angle $\Delta\theta$ versus communication distance for all LEDs. Two trends are clear: (1) For each LED, the optimal $\Delta\theta$ grows with distance, as longer distances reduce pixels per bit, increasing ISI susceptibility. (2) At fixed distance, outer LEDs need smaller $\Delta\theta$ since their longer trails per angle mitigate ISI. In essence, as distance increases, the blur scale grows, and thus the optimal control angle also increases. We summarize this as a scaling insight: $\Delta\theta^\star \propto \frac{\sigma_{\mathrm{eff}}(D)}{r}.$ This gives a clear design rule, linking the observed trend to a geometry–blur tradeoff.}

\subsubsection{Throughput}
\rev{We define throughput as the number of bits received per unit of time under ideal conditions. In this baseline, we assume uncoded transmission, no guard angles, no inter-LED interference, perfect synchronization, and constant rotation speed. With three rotations per unit time, the data rate is $3\cdot2\pi/\Delta\theta$ [bps].
Fig.~\ref{fig:throughput} shows the throughput for different $\Delta\theta$ for the innermost LED ($i=1$) at various communication distances. The resulting $\Delta\theta$ is locally optimal within this model. 
While the current setup provides only a modest throughput, the achievable rate can be significantly increased by using higher-speed motors and higher frame rate image sensors.
Future work will consider coding, guard angles, synchronization errors, and multi-LED effects.}

\begin{figure}[t] 
 \vspace{-4mm}
 \centering
 \includegraphics[width= 0.9\linewidth]{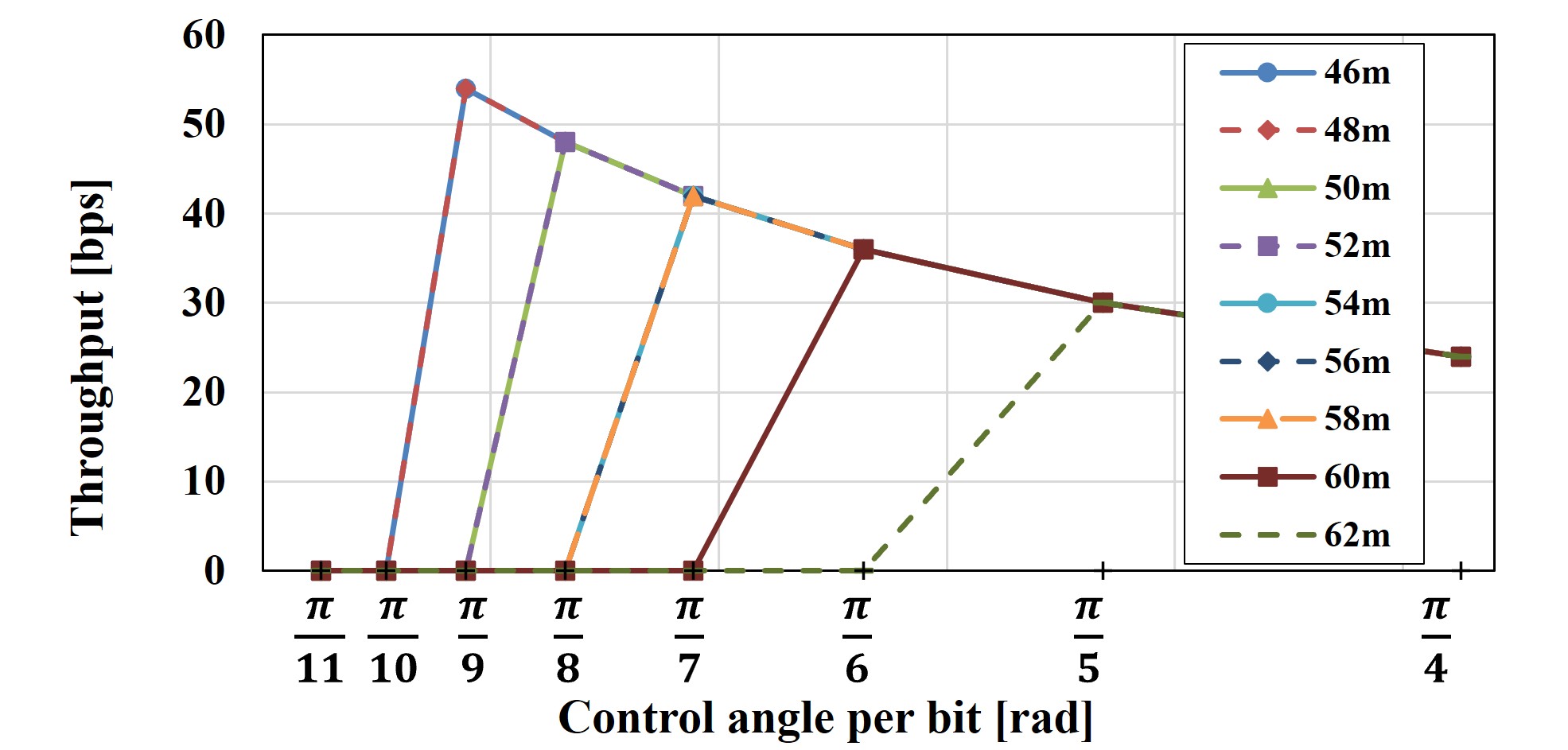}
  \vspace{-2mm}
 \caption{Throughput vs. control angle per bit ($\Delta\theta$).}
 \label{fig:throughput}
 \vspace{-6mm}
\end{figure} 

% 図に$i=1$のLEDにおける各通信距離に対する最適な$\Delta\theta$の値を示す．通信距離が長くなると，光跡が投影されるピクセル数の減少によりISIが生じやすくなり，必要となる制御角度は大きくなる．
% Fig.~\ref{fig:distancevsAngle} shows the control angle 
% $\Delta\theta$ versus communication distance. We observe that the required control angle increases as the distance increases. 

\vspace{-3mm}
\section{Conclusion}
We developed an analytical model for light-trail ISC with a rotating-pattern transmitter. By focusing on adjacent-trail ISI and adopting the hardest-pair midpoint threshold, we obtained a closed-form BER that matches Monte Carlo/measurements across tested settings, indicating non-adjacent trails are negligible for performance prediction. 
The analysis clarifies how distance, blur, and the per-bit control angle $\Delta\theta$ shape BER and provides a practical rule to select $\Delta\theta$. We demonstrated throughput maximization under an error-free target. 
Future work will relax the adjacent-only assumption (larger $K$), cover stronger blur/longer ranges/off-axis geometry, and explore decision rules that retain interpretability and efficiency.

\vspace{-3mm}
\bibliographystyle{IEEEtran}
\bibliography{refs}

\end{document}